\documentclass[journal]{IEEEtran}
\ifCLASSINFOpdf
\else
\fi
\usepackage{amsthm}
\usepackage{amssymb}
\usepackage{amsmath}
\usepackage{times}
\usepackage{graphicx}
\usepackage{graphics}
\usepackage{subfigure}
\usepackage{multirow}
\usepackage{caption}
\usepackage{makecell}
\usepackage{cite}
\usepackage{color}
\usepackage{booktabs}
\usepackage{epstopdf}
\usepackage{cite}
\usepackage{bm}
\usepackage[noend]{algpseudocode}
\usepackage{algorithmicx,algorithm}

\newtheorem{theorem}    {Theorem}
\newtheorem{definition} {Definition}

\newtheorem{lemma}      {Lemma}
\newtheorem{remark}     {Remark}
\newtheorem{example}    {Example}

\begin{document}
%
\title{
Optimal Memory Scheme for Accelerated Consensus Over Multi-Agent Networks}
%
%
%

\author{Jiahao~Dai,
        Jing-Wen~Yi$^*$,
        Li~Chai,~\IEEEmembership{Member,~IEEE,}

\thanks{This work was supported by the National Natural Science Foundation of China (grant 61625305 and 61701355).}
\thanks{$^*$ Corresponding author. Jing-Wen Yi is with the Engineering Research Center of Metallurgical Automation and Measurement Technology, Wuhan University of Science and Technology, Wuhan, China. e-mail: yijingwen@wust.edu.cn.}
\thanks{Jiahao Dai is with the same University. e-mail: daijiahao@wust.edu.cn.}
\thanks{Li Chai is with the same University. e-mail: chaili@wust.edu.cn.}}
%
%

\markboth{Journal of \LaTeX\ Class Files,~Vol.~xx, No.~x, August~202x}%
{Shell \MakeLowercase{\textit{et al.}}: Bare Demo of IEEEtran.cls for IEEE Journals}
%



\maketitle

\begin{abstract}
The consensus over multi-agent networks can be accelerated by introducing agent's memory to the control protocol.
In this paper, a more general protocol with the node memory and the state deviation memory is designed.
We aim to provide the optimal memory scheme to accelerate consensus.
The contributions of this paper are three:
(i) For the one-tap memory scheme,  we demonstrate that the state deviation memory is useless for the optimal convergence.
(ii) In the worst case, we prove that it is a vain to add any tap of the state deviation memory, and the one-tap node memory is sufficient to achieve the
optimal convergence.
(iii) We show that the two-tap state deviation memory is effective on some special networks, such as star networks.
Numerical examples  are listed to illustrate the validity and correctness of the obtained results.

\end{abstract}


%
\IEEEpeerreviewmaketitle

\section{Introduction}

Consensus is a basic problem in distributed coordination control over multi-agent networks, which has been studied extensively in the last
decades \cite{2004Consensus,2007Consensus,2007A,2013An}.
The usual idea to solve this problem is to use the local information to design a consensus protocol, so that the states of all agents can reach a common value over time.
Due to the differences in network topologies, control strategies and system models, scholars study the problem of consensus from a multi-faceted perspective, including consensus in switching topology \cite{2004Consensus,6596518}, consensus with communication delays \cite{4639466,2018Consensus}, consensus with sampling data \cite{7500132,8664588}, quantized consensus \cite{6862838,7307996}, consensus in nonlinear systems \cite{2013Consensus,7094259}, etc.

Convergence rate is an important performance indicator of consensus.
Optimizing the weight of the network \cite{XIAO200465,4627467,2011Network,7389373} is a effective method
to accelerate the consensus.
Xiao et al. \cite{XIAO200465} concerned on how to choose the network weights to derive the fastest convergence rate, and gave the optimal constant edge weights to achieve the accelerated consensus.
To improve the convergence rate, Kokiopoulou et al. \cite{4627467} applied a polynomial filter on the network matrix to shape its spectrum.
You et al. \cite{2011Network} revealed how the network affect the consensus of the discrete-time multi-agent system, and gave a lower bound of the optimal convergence rate.
Apers et al. \cite{7389373} proposed a preconditioner to optimize the edge weights of a given graph and cluster the eigenvalues towards better polynomial acceleration.

Some researchers consider time-varying control strategies \cite{6338354,KIBANGOU201419,6876198,2019Average} to improve the convergence rate.
By analyzing the properties of Chebyshev polynomials,  Montijano et al. \cite{6338354} designed a fast and stable distributed consensus algorithm.
Kibangou \cite{KIBANGOU201419} considered the multi-agent system under time-varying control,
and applied a matrix factorization approach to get the condition of finite-time consensus.
Safavi and Khan \cite{6876198} introduced an approach termed as successive nulling of eigenvalues under time-varying control,
and proposed some necessary and sufficient conditions for the multi-agent system to achieve the finite-time consensus.
Yi et al. \cite{2019Average} studied the fast consensus problem from the perspective of graph filtering,
and provided some explicit formulas for the optimal convergence rate by using the period control strategy.

The method of using agent's past information is also effective in improving the convergence rate \cite{5411823,8716798,IROFTI2020108776,Yi2021convergence,16M1076629,8431817,8795623,9117158}.
Oreshkin et al. \cite{5411823} proposed a short node memory scheme to accelerate the convergence in distributed averaging.
Kia et al. \cite{8716798} introduced agent's memory into robust dynamic average consensus algorithms, and got the optimal convergence rate by using the method of root locus.
Irofti \cite{IROFTI2020108776} intentionally introduced agent's past information in the control protocol to accelerate the consensus, and derived an optimized value of the convergence rate.
Yi et al. \cite{Yi2021convergence} considered the control protocol with the node memory, and gave the optimal worst-case convergence rate for an uncertain graph set.
These researches all uses agent's own past state stored in memory to accelerate the consensus, which is called the node memory scheme in this paper.
To improve the convergence rate, some researchers also utilize neighbours' past information stored in memory, which is called the state deviation memory.
Olshevsky \cite{16M1076629} added the state deviation memory to the consensus protocol, and get the linear convergence related to the number of nodes.
Bu et al. \cite{8431817} considered a consensus protocol with the state deviation memory, and provided a convergence bound in the worst case.
Moradian et al. \cite{8795623} applied neighbours' past information to achieve the fast consensus of a single integrator system, and determined the depth of memory that can improve the convergence rate.
Pasolini et al. \cite{9117158} proposed a general protocol with the state deviation memory, and formulated the optimal convergence rate when adding the one-tap memory.

Although the aforementioned results have great improvement for the convergence rate of consensus, it's worth noting that there still leave some problems that have not been considered.
For example, in the consensus protocol, introducing the node memory or the state deviation memory can accelerate the consensus, so which type of memory has the better acceleration performance?
If both two memory schemes are introduced to the control protocol, will the consensus be achieved faster?
In this paper, a more general control protocol with memory is designed to solve these problems.
The main contributions can be summarized as follows.

(i) The optimal convergence rate of the short memory scheme is formulated by using Jury stability criterion.
It is found that the one-tap state deviation memory is useless for the optimal convergence.

(ii) It is proved that the optimal worst-case convergence rate cannot be improved by adding more than one-tap memory by transforming the optimization problem of the convergence rate into the robust stabilization of the feedback system.
In the worst case, the one-tap node memory scheme is  sufficient to achieve the optimal convergence, and
it is a vain to add any tap of the state deviation memory.

(iii) For star networks, an optimized convergence rate of the two-tap memory scheme is given,
which indicates that the two-tap state deviation memory is effective on some special networks.

The remainder of this paper is organized as follows.
In Section II, the problem statement is described.
Section III introduces the optimal short memory scheme.
In Section IV, the optimal worst-case memory scheme is proposed, and a special case on star networks is given.
Finally, Section V concludes this paper.

\section{Problem statement}

\subsection{Graph Theory}
Let $\mathcal{V}\!=\!\{\text{v}_1,\text{v}_2,\cdots,\text{v}_N\}$ be the set of agents or nodes,
$\mathcal{E \!\subseteq\! V\times V}$ be a set of edges, and $\mathcal{A}\!=\![a_{ij}]\!\in\! \mathbb{R}^{N\!\times\! N}$ be a weighted adjacency matrix.
The interactions among agents are modeled as an undirected network $\mathcal{G\!=\!(V,E,A)}$.
The edge between $\text{v}_i$ and $\text{v}_j$ is denoted by $e_{ij}\!=\!(\text{v}_{i},\text{v}_{j})\!\in\! \mathcal{E}$,
indicating that there is a communication link between $\text{v}_i$ and $\text{v}_j$.
The adjacency element $a_{ij}\!=\!a_{ji}\!>\!0$ if $e_{ij}\in\mathcal{E}$.
Let ${\mathcal{N}_i} \!= \!\left\{ {{\text{v}_j} \in \mathcal{V}:({\text{v}_i},{\text{v}_j}) \in \mathcal{E} } \right\}$
be the set of neighbors of agent $\text{v}_i$.
The degree of $\text{v}_{i}$ are represented by ${d_i}\!= \!\sum\nolimits_{j = 1}^N {{a_{ij}}}$.
Define the Laplacian matrix of $\mathcal{G}$ as $L\!=\!\mathcal{D\!-\!A}$, where $\mathcal{D}\!:=\!diag\{d_{1},\cdots,d_{N}\}$.
For a connected network, all the eigenvalues of $L$ are real in an ascending order as $0 \!=\! {\lambda _1} \!<\! {\lambda _2} \!\le\!  \cdots  \!\le\! {\lambda _N}$.
\begin{lemma}
\cite{2004Consensus} For any connected undirected network $\mathcal{G}$, its Laplacian matrix is positive-semidefinite,
and has the decomposition ${L=V\Lambda V^{T}}$, where $\Lambda\!=\!diag\{\lambda_{1},\lambda_{2},\cdots,\lambda_{N}\}$ and $V\!=\![\bm{v}_{1},\bm{v}_{2},\cdots,\bm{v}_{N}]\!\in\!\mathbb{R}^{N\times N}$ are unitary.
Zero is a single eigenvalue of $L$, and the corresponding eigenvector is $v_{1}\!=\!\frac{1}{\sqrt{N}}\bm{1}$, where $\bm{1}$ denotes the vector of all ones.
\end{lemma}
\subsection{Problem statement}
The discrete-time dynamics of the agent $\text{v}_i$ is given by
\begin{equation}
{x_i}(k + 1) = {x_i}(k) + {u_i}(k)
\end{equation}
where $x_i(k)\!\in\! \mathbb{R}$ denotes the state and $u_i(k)\!\in\!\mathbb{R}$ denotes the control input.

The control protocol is designed as
\begin{equation}
\begin{aligned}
{u_i}(k) =& \sum\limits_{m = 0}^M {{\varepsilon _m}\sum\limits_{j \in {\mathcal{N}_i}} {{a_{ij}}({x_j}(k - m) - {x_i}(k - m)} )}
\\&+ \sum\limits_{m = 0}^M {{\theta _m}{x_i}(k - m)},
\end{aligned}
\end{equation}
where ${\varepsilon _m}\in\mathbb{R}^+,{\theta _m}\in\mathbb{R}$ are control parameters.
Each agent updates its state by using the node memory and the state deviation memory.
The initial states are set as
\[{x_i}(-M) =  \cdots  = {x_i}( - 1) = {x_i}(0).\]
Note that the control protocol (2) is general.
Consider the following two examples.
\\(i) If $\varepsilon_1,\ldots,\varepsilon_M=0$, then (2) becomes the protocol with only node memory
\begin{equation} \nonumber
\begin{aligned}
{u_i}(k) = \varepsilon_0 \sum\limits_{j \in {\mathcal{N}_i}} {{a_{ij}}({x_j}(k) - {x_i}(k)} )
        + \sum\limits_{m = 0}^M {{\theta _m}{x_i}(k - m)},
\end{aligned}
\end{equation}
which has been investigated in \cite{Yi2021convergence} by using the stability theory.
\\(ii) If $\theta_0,\ldots,\theta_M=0$, then (2) becomes the protocol with only state deviation memory
\begin{equation} \nonumber
{u_i}(k) = \sum\limits_{m = 0}^M {{\varepsilon _m}\sum\limits_{j \in {\mathcal{N}_i}} {{a_{ij}}({x_j}(k - m) - {x_i}(k - m)} )},
\end{equation}
which has been investigated in \cite{9117158} by using FIR filtering.


\begin{definition}
The average consensus is said to be reached asymptotically if
\begin{equation}
\mathop {\lim }\limits_{k \to \infty } {x_i}(k) = \frac{1}{N}\sum\nolimits_{j = 1}^N {{x_j}(0)}= \bar x ,\,\,i = {1, \ldots ,N}
\end{equation}
holds for any initial state $x_i(0)$.
\end{definition}

\subsection{Convergence rate}

The system (1) under the control protocol (2) can be written as
\begin{equation}
\bm{x}(k + 1) = [(1 + {\theta _0})I - {\varepsilon _0}L]\bm{x}(k) + \sum\limits_{m = 1}^M {({\theta _m}I \!-\! {\varepsilon _m}L)\bm{x}(k \!-\! m)},
\end{equation}
where $\bm{x}(k)$ is the column stack of $x_i(k), i=1,\ldots,N$.

\begin{lemma}
The consensus of system (4) is achieved only if $\sum\limits_{m = 0}^M {{\theta _m} = 0}$.
\end{lemma}

\begin{proof}
When ${k \!\to\! \infty }$, the stationary solutions of (4) satisfies
\begin{equation}
{\bar x} \bm{1} =   {\bar x}[(1 \!+\! {\theta _0})I \!-\! {\varepsilon _0}L]\bm{1}  \!+\! {\bar x}\sum\limits_{m = 1}^M {({\theta _m}I \!-\! {\varepsilon _m}L)\bm{1}}.
\end{equation}
Note that $L\bm{1} = \bm{0}$.
The stationary solutions in equation (5) are kept only if $\sum\limits_{m = 0}^M {{\theta _m} = 0}$.
Thus, the condition $\sum\limits_{m = 0}^M {{\theta _m} = 0}$  is required to ensure that the consensus can be reached.
\end{proof}


Perform the graph Fourier transform \cite{6808520}
\[
{\hat x}_i(k)=\bm{v}_i^T\bm{x}(k),i=1,\dots,N,
\]
and get
\begin{equation}
\hat x_i(k + 1) = (1 + {\theta _0} - {\varepsilon _0}\lambda_i )\hat x_i(k) + \sum\limits_{m = 1}^M {({\theta _m} \!-\! {\varepsilon _m}\lambda_i )} \hat x_i(k - m).
\end{equation}
Then the agent's state signal can be analyzed in the graph spectral domain.

\begin{lemma}
Assume that $\sum\limits_{m = 0}^M {{\theta _m} = 0}$.
The consensus of system (4) is achieved if and only if
$\mathop {\lim }\limits_{k \to \infty } {{\hat x}_i}(k) = 0$ holds for any $i\in\{2,3,\ldots,N\}$.
\end{lemma}

\begin{proof}
For a connected graph, $\lambda_1=0$ and
\[
{{\hat x}_1}(k + 1) = {{\hat x}_1}(k) + \sum\limits_{m = 0}^M {{\theta _m}} {{\hat x}_1}(k - m).
\]
Since ${\hat x_i}(-M) \!= \!  \cdots  \! =  \!{\hat x_i}( - 1)  \!= \! {\hat x_i}(0)$ and $\sum\limits_{m = 0}^M {{\theta _m}  \!= \! 0}$,
we have
\[
{{\hat x}_1}(k) = {{\hat x}_1}(0), \,\,\,\,\forall k \ge 0.
\]
It follows that
\[
{\bm{v}_1}{\hat x}_1(k) ={\bm{v}_1}{{\hat x}_1}(0) =\frac{1}{N}{\bm{1}\bm{1}^T}\bm{x}(0)= \bar x \bm{1}.
\]
Then the state of system (3) can be written as
\begin{equation}
\mathop {\lim }\limits_{k \to \infty } \bm{x}(k) = \bar x \bm{1} + \mathop {\lim }\limits_{k \to \infty } \sum\limits_{i = 2}^N {{\bm{v}_i}{{\hat x}_i}(k)}.
\end{equation}
Substituting $\mathop {\lim }\limits_{k \to \infty } {{\hat x}_i}(k) \!=\! 0, i \!= \!2,\ldots,N$ into (7), the sufficiency can be proved directly.
Suppose that there is a scalar $j\!\in\!\{2,\ldots,N\}$ that satisfies $\mathop {\lim }\limits_{k \to \infty } {{\hat x}_j}(k)\!=\! \delta  \!\ne\! 0$.
Since $\left\langle {{\bm{v}_i},{\bm{v}_j}} \right\rangle  = 0$ holds for any $i\ne j$, then
$\mathop {\lim }\limits_{k \to \infty } \bm{x}(k)- \bar x \bm{1} \!\ne\! \bm{0}$.
This contradiction proves the necessity.
\end{proof}

Denote $\bm{y}_i(k)=[{{\hat x}_i}(k),{{\hat x}_i}(k-1),\cdots,{{\hat x}_i}(k-M)]^T$.
The problem of consensus is transformed into the simultaneous stability problem of $N\!-\!1$ systems
\begin{equation}
\bm{y}_i(k + 1) = \Gamma(\lambda_i) \bm{y}_i(k), i = 2,3,\ldots,N,
\end{equation}
where
\begin{equation}
\Gamma(\lambda_i)  \!=\! \left[ {\begin{array}{*{20}{c}}
{1 \!+\! {\theta _0} \!-\! {\varepsilon _0}\lambda_i}\,&\,{{\theta _1} \!-\! {\varepsilon _1}\lambda_i}& \cdots &{{\theta _M} \!-\! {\varepsilon _M}\lambda_i}\\
1&0& \cdots &0\\
0& \ddots & \ddots & \vdots \\
0&0&1&0
\end{array}} \right].
\end{equation}

\begin{definition}
Assume that the system (4) on a given network $\mathcal{G}$.
Define the convergence rate of the consensus as \cite{XIAO200465,4627467}
\begin{equation}
{r_M} = \mathop {\max }\limits_{i \in \{ 2, \ldots N\}} {\rho}(\Gamma(\lambda_i) ),
\end{equation}
where ${\rho } \left( \cdot\right)$ denotes the spectral radius.
\end{definition}

This paper focuses on designing the control parameters ${\varepsilon _0}, \ldots ,{\varepsilon _M}$ and ${\theta _0}, \ldots ,{\theta _M}$ to achieve the consensus with fast convergence speed, and then obtain the optimal memory scheme.

\section{Optimal short memory Scheme on a given network}

This section aims to analyze the optimal short memory scheme on a given network.

\begin{lemma} (Schur Complement \cite{boyd2004convex}) Given a matrix
\[W = \left[ {\begin{array}{*{20}{c}}W_1&W_2\\W_3&W_4\end{array}}\right],\]
 with nonsigular $W_1\in {\mathbb{R}^{\mu \times \mu}}$, $W_2 \in {\mathbb{R}^{N \times \mu}}$, $W_3 \in {\mathbb{R}^{\mu \times N}}$, and $W_4 \in {\mathbb{R}^{N \times N}}$. Then
$\det W = \det W_1 \cdot \det (W_4 - W_3{W_1^{ - 1}}W_2)$.
\end{lemma}

\begin{lemma} (Jury Stability Criterion \cite{4066881})
Given a second-order polynomial $d(z)\!=\!z^2+a_1z+a_0$.
The roots of $d(z)=0$ are all in the unit circle if and only if
\[d(1) > 0,d( - 1) > 0,\left| {{a_0}} \right| < 1.\]
\end{lemma}

Define ${p}(z,\lambda_i) \!=\! \det (zI - \Gamma(\lambda_i) )$.
Using the method of Shur complement,
the characteristic polynomial of $\Gamma(\lambda_i)$ can be easily calculated as
\begin{equation}
\begin{aligned}
{p}(z,\lambda_i)= &{z^{M + 1}} + ({\varepsilon _0}{\lambda _i} - 1 - {\theta _0}){z^M}
           \\&+ \sum\limits_{m = 1}^M {({\varepsilon _m}{\lambda _i} - {\theta _m}){z^{M - m}}}.
\end{aligned}
\end{equation}
Then the problem of accelerated consensus can be converted into the optimization problem
\begin{equation}
\begin{array}{c}
\mathop {\min }\limits_{{\varepsilon _0}, \ldots ,{\varepsilon _M},{\theta _0}, \ldots ,{\theta _M}} r_M\\
\begin{array}{*{20}{c}}
{s.t.}&{\begin{array}{l}
{p}(z,\lambda_i) = 0 , \,i=2,\ldots,N.
\end{array}}
\end{array}
\end{array}
\end{equation}

The optimization problem (12) is difficult to solve especially when $M$ is large.
The optimal convergence rate with $M=0$ is $r_0^*=\frac{{{\lambda _N} - {\lambda _2}}}{{{\lambda _N} + {\lambda _2}}}$,
which has been solved in \cite{XIAO200465}.
In this section, we explore the optimal convergence rate when $M=1$.

\begin{theorem}
Consider the consensus of the system (4) on a connected network $\mathcal{G}$.
The optimal convergence rate of $M\!=\!1$ is
\begin{equation}
{r^*_1} = \frac{{\sqrt {{\lambda _N}}  - \sqrt {{\lambda _2}} }}{{\sqrt {{\lambda _N}}  + \sqrt {{\lambda _2}} }}
\end{equation}
with the optimal control parameters
\begin{equation}
\begin{array}{l}
{\varepsilon _0^*} = \frac{4}{{{{(\sqrt {{\lambda _N}}  + \sqrt {{\lambda _2}} )}^2}}},{\varepsilon _1^*} = 0,\\
{\theta _0^*} = {\left( {\frac{{\sqrt {{\lambda _N}}  - \sqrt {{\lambda _2}} }}{{\sqrt {{\lambda _N}}  + \sqrt {{\lambda _2}} }}} \right)^2},{\theta _1^*} =  - {\theta _0}.
\end{array}
\end{equation}
\end{theorem}

\begin{proof}
When $M\!=\!1$, the system can be written as
\[
\bm{x}(k + 1) = [(1 + {\theta _0})I - {\varepsilon _0}L]\bm{x}(k) + {({\theta _1}I \!-\! {\varepsilon _1}L)\bm{x}(k \!-\! 1)}.
\]
The consensus is achieved if and only if the roots of
\begin{equation}
\begin{aligned}
{p}(z,\lambda_2) \!=&{z^2} \!+\! ({\varepsilon _0}{\lambda _2} \!-\! 1 \!-\! {\theta _0})z \!+\! {\varepsilon _1}{\lambda _2}\! +\! {\theta _0}\!=0,\\
{p}(z,\lambda_N) \!=&{z^2} \!+\! ({\varepsilon _0}{\lambda _N} \!-\! 1 \!-\! {\theta _0})z \!+\! {\varepsilon _1}{\lambda _N}\! +\! {\theta _0}\!=0\\
\end{aligned}
\end{equation}
are in the unit circle.
Let $z = r\tilde z$ in (15).
Then the roots of (15) are in the circle with radius $r$ if and only if the roots of
\begin{equation}
\begin{array}{*{20}{l}}
{r^2}{{\tilde z}^2} + ({\varepsilon _0}{\lambda _2} - {\theta _0} - 1)r\tilde z + {\varepsilon _1}{\lambda _2} + {\theta _0} = 0,\\
{r^2}{{\tilde z}^2} + ({\varepsilon _0}{\lambda _N} - {\theta _0} - 1)r\tilde z + {\varepsilon _1}{\lambda _N} + {\theta _0} = 0
\end{array}
\end{equation}
are in the unit circle.
According to the Jury stability criterion, the roots of (16) are in or on the unit circle if and only if
\begin{subequations}
\begin{align}
&{r^2} + ({\varepsilon _0}{\lambda _2} - {\theta _0} - 1)r + {\varepsilon _1}{\lambda _2} + {\theta _0} \ge 0\\
&{r^2} + ({\varepsilon _0}{\lambda _N} - {\theta _0} - 1)r + {\varepsilon _1}{\lambda _N} + {\theta _0} \ge 0\\
&{r^2} - ({\varepsilon _0}{\lambda _2} - {\theta _0} - 1)r + {\varepsilon _1}{\lambda _2} + {\theta _0} \ge 0\\
&{r^2} - ({\varepsilon _0}{\lambda _N} - {\theta _0} - 1)r + {\varepsilon _1}{\lambda _N} + {\theta _0} \ge 0\\
&{r^2} - {\varepsilon _1}{\lambda _2} - {\theta _0} \ge 0,\\
&{r^2} - {\varepsilon _1}{\lambda _N} - {\theta _0} \ge 0.
\end{align}
\end{subequations}
To eliminate $\theta_0$, we multiply ${1 \!-\! r}$ to (17e) and add (17a), and give
\begin{equation}
r{\lambda _2}({\varepsilon _0} +{\varepsilon _1}) - r{(r\! -\! 1)^2} \ge 0.
\end{equation}
Similarly, multiply ${1 \!+\! r}$ to (17f) and add (17d), give
\begin{equation}
- r{\lambda _N}({\varepsilon _0} + {\varepsilon _1}) + r{(r \!+\! 1)^2} \ge 0.
\end{equation}
To eliminate $\varepsilon _0$ and $\varepsilon _1$, we first multiply $\lambda_N$ to (18) and  $\lambda_2$ to (19), and have
\begin{equation}
r{\lambda _2}{\lambda _N}({\varepsilon _0} + {\varepsilon _1}) - {\lambda _N}r{(r - 1)^2} \ge 0,
\end{equation}
\begin{equation}
 - r{\lambda _2}{\lambda _N}({\varepsilon _0} + {\varepsilon _1}) + {\lambda _2}r{(r + 1)^2} \ge 0.
\end{equation}
Then add (20) and (21), get
\[
{\lambda _2}{(r + 1)^2} - {\lambda _N}{(r - 1)^2} \ge 0.
\]
It follows that
\begin{equation}
r \ge \frac{{\sqrt {{\lambda _N}}  - \sqrt {{\lambda _2}} }}{{\sqrt {{\lambda _N}}  + \sqrt {{\lambda _2}} }}.
\end{equation}
The optimal solution $r=r_1^*$ is obtained if and only if
\begin{equation}
\left[ {\begin{array}{*{20}{c}}
{{\lambda _2}r}&{{\lambda _2}}&{1 \!-\! r}\\
{ - {\lambda _N}r}&{{\lambda _N}}&{1 \!+\! r}\\
0&{ - {\lambda _2}}&{ - 1}\\
0&{ - {\lambda _N}}&{ - 1}
\end{array}} \right] \cdot
\left[ {\begin{array}{*{20}{c}}
{{\varepsilon _0}}\\
{{\varepsilon _1}}\\
{{\theta _0}}
\end{array}} \right] = \left[ {\begin{array}{*{20}{c}}
{\!-{r^2}\! +\!r}\\
{\!-{r^2} \!-\! r}\\
{\!-{r^2}}\\
{\!-r^2}
\end{array}} \right].
\end{equation}
The solution of equation (23) is given by (14).
It is verified that the remaining constraints in (17) are satisfied.
This completes the proof.
\end{proof}

\begin{remark}
The optimal convergence rate and the corresponding control parameters are related to the eigenratio $\lambda_2/\lambda_N$ of the Laplacian matrix.
The larger eigenratio corresponds to the better network connectivity, which leads to the faster convergence of consensus.
\end{remark}

\begin{remark}
By utilizing time-varying control without memory,
\cite{2019Average} proposed that the lower bound of the optimal convergence rate is $\frac{{\sqrt {{\lambda _N}}  - \sqrt {{\lambda _2}} }}{{\sqrt {{\lambda _N}}  + \sqrt {{\lambda _2}} }}=r_1^*$.
This means that the optimal convergence rate of the constant control scheme with the one-tap memory has reached the limit value of the convergence rate of the time-varying control scheme without memory.
\end{remark}

\begin{remark}
Note that $\varepsilon_1^*=0$ in (14).
This means that it is not necessary to add the state deviation memory when $M\!=\!1$.
The optimal convergence rate with the one-tap state deviation memory proposed in \cite{9117158} is
$\frac{{{\lambda _N} - {\lambda _2}}}{{{\lambda _N} + 3{\lambda _2}}}>r_1^*$, which also corroborates our analysis.
\end{remark}

The following two statements are worth noting.
(i) If only the state deviation memory is used, the convergence rate can be improved as $M$ increases \cite{9117158}.
(ii) If only the node memory is used, the convergence rate cannot be improved for any $M>1$ in the worst case \cite{Yi2021convergence}.
Then a question arises spontaneously: can the convergence rate of the general memory scheme in this paper be improved as $M$ increases?
This question is explored in the next section.

\section{The worst-case optimal memory scheme and a special case}

This section proposes the optimal memory scheme in the worst case, and introduces a special case on star networks.



\begin{lemma} (Routh-Hurwitz Criterion \cite{1100537})
Given a third-order polynomial $d(s)\!=\!a_0s^3+a_1s^2+a_2s+a_3$.
The roots of $d(s)=0$ are all in the open left half plane (the polynomial $d(s)$ is stable) if and only if
\[a_0 > 0,a_1 > 0,a_2 > 0,a_3 > 0,a_1a_2-a_0a_3 > 0.\]
\end{lemma}

Assume that the system (4) on the uncertain network $\mathcal{G}_{[\alpha ,\beta ]}$, where $\mathcal{G}_{[\alpha ,\beta ]}$ represents all the networks with the nonzero eigenvalues of $L$ in the interval $[\alpha, \beta]$.

Define the worst-case convergence rate of the consensus as
\begin{equation}
{\gamma _M} = \mathop {\sup }\limits_{\mathcal{G} \in {\mathcal{G}_{[\alpha ,\beta ]}}} {r_M}.
\end{equation}

\begin{lemma}
The worst-case convergence rate satisfies
\begin{equation}
{\gamma _M} = \mathop {\sup }\limits_{\mathcal{G} \in {\mathcal{G}_{[\alpha ,\beta ]}}} \rho \left( {\Gamma (\lambda )} \right) = \mathop {\sup }\limits_{\mathcal{G} \in {\mathcal{G}_{[\alpha ,\beta ]}}} \mathop {\max }\limits_{n \in \{ 1,2, \ldots ,M + 1\} } \left| {{z_n}(\lambda )} \right|,
\end{equation}
where ${{z}_{n}(\lambda)}$ denotes the root of ${p}(z,\lambda)=0$, and ${p}(z,\lambda)$ is defined in (11).
\end{lemma}

From Theorem 1, the optimal worst-case convergence rate of $M=1$ is
\begin{equation}
{\gamma_{1}^*} = \frac{{\sqrt {{\beta}}  - \sqrt {{\alpha}} }}{{\sqrt {{\beta}}  + \sqrt {{\alpha}} }},
\end{equation}
and the corresponding control parameters are
\begin{equation}
\begin{array}{l}
{\varepsilon _0^*} = \frac{4}{{{{(\sqrt {{\beta}}  + \sqrt {{\alpha}} )}^2}}},{\varepsilon _1^*} = 0,\\
{\theta _0^*} = {\left( {\frac{{\sqrt {{\beta}}  - \sqrt {{\alpha}} }}{{\sqrt {{\beta}}  + \sqrt {{\alpha}} }}} \right)^2},{\theta _1^*} =  - {\theta _0}.
\end{array}
\end{equation}
Before giving the result, we make some settings.
Let
\begin{equation} \nonumber
\begin{aligned}
p(rz,\lambda )= &{r^{M + 1}}{z^{M + 1}} + ({\varepsilon _0}\lambda  - 1 - {\theta _0}){r^M}{z^M}
           \\&+ \sum\limits_{m = 1}^M {({\varepsilon _m}\lambda  - {\theta _m}){r^{M - m}}{z^{M - m}}}.
\end{aligned}
\end{equation}
Set up the feedback system as shown in Fig.1, where
\begin{equation}  \nonumber
P(rz,\lambda ) = \frac{{{(rz)^{ - 1}}}}{{1 \!-\! {(rz)^{ - 1}}}}\lambda ,C(rz) = \frac{{\sum\limits_{m = 0}^M {{\varepsilon _m}{r^{ - m}}{z^{ - m}}} }}{{1 \!-\! \sum\limits_{m = 1}^M {\sum\limits_{j = 0}^{m - 1} {{\theta _j}{r^{ - m}}{z^{ - m}}} } }}.
\end{equation}
The transfer function of the closed-loop system is
\begin{equation}  \nonumber
G(rz,\lambda ) = \frac{{P(rz,\lambda )}}{{1 + P(rz,\lambda )C(rz)}} = \frac{{\lambda \sum\limits_{m = 0}^M {{r^{M - m}}{z^{M - m}}} }}{{p(rz,\lambda )}}.
\end{equation}

\begin{figure}[!htbp]
\centering
\includegraphics[height=2.2cm]{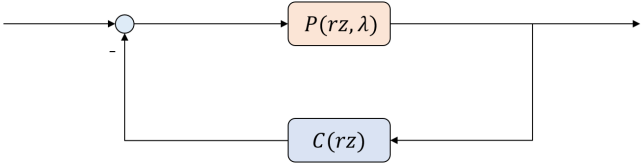}
\caption{The feedback system}
\label{fig:label}
\end{figure}

Next,
by converting the problem of fast consensus to the robust stabilization of the feedback system,
we give the optimal worst-case convergence rate of any $M\ge 1$ as follows.

\begin{theorem}
Consider the consensus of system (4) on the uncertain network $\mathcal{G}_{[\alpha ,\beta ]}$.
The optimal worst-case convergence rate is
\begin{equation}
{\gamma_{M}^*} = \frac{{\sqrt {{\beta}}  - \sqrt {{\alpha}} }}{{\sqrt {{\beta}}  + \sqrt {{\alpha}} }},
\end{equation}
and a set of optimal control parameters is given by (27).
\end{theorem}

\begin{proof}
The roots of $p(z,\lambda)\!=\!0$ are in the circle with radius $r$ if and only if the polynomial $p(rz,\lambda)$ is stable.
This means that the system with the transfer function ${G}(rz,\alpha)$ needs to be stable for any $\lambda \in [\alpha ,\beta ]$ to ensure the worst-case convergence rate $r$.
The statement that ${G}(rz,\alpha)$ is stable for any $\lambda \in [\alpha ,\beta ]$ can be regarded as the problem of gain margin optimization \cite{1992Feedback} (Chapter 11), that is, a stable ${G}(rz,\alpha)$ with the gain margin $k_{\sup}\!\ge\! \beta /\alpha$.
Then the gain margin of $P(rz,\alpha)$ needs to satisfy $k_{\sup}\!\ge\! \beta /\alpha$.
Note that the optimal gain margin of $P(rz,\alpha)=\frac{{{(r\tilde z)^{ - 1}}}}{{1 - {(r\tilde z)^{ - 1}}}}\alpha $ has been given by
$k_{\sup}\!=\!{({\frac{{1 + r}}{{1 - r}}})^2}$ \cite{Yi2021convergence} (Lemma 5).
It follows from ${({\frac{{1 + r}}{{1 - r}}})^2}\!\ge\! \beta /\alpha$ that
$ r \ge \frac{{\sqrt {{\beta}}  - \sqrt {{\alpha}} }}{{\sqrt {{\beta}}  + \sqrt {{\alpha}} }}$.
A set of control parameters to reach the optimal worst-case convergence rate $r={\gamma_{M}^*}$ is given by (27).
\end{proof}


\begin{remark}
Note that $\gamma_{M}^*=\gamma_{1}^*$ for any $M\ge 1$.
This means that,
in the worst case, the one-tap node memory is sufficient to achieve the optimal convergence, and applying the state deviation memory in \cite{9117158} is unnecessary.
\end{remark}

It is worth noting that \cite{Yi2021convergence} has proposed that the convergence rate of the two-tap node memory scheme cannot be improved under any given network.
However, we find that the two-tap state deviation memory is effective on some special networks.
In particular, on star networks, an explicit formula for the convergence rate with $M\!=\!2$ is given as follows.

\begin{theorem}
Consider the MAS (1) on star networks with $N$ nodes under the control protocol (2).
Denote $M\!=\!2$ and $\mu \!=\! \frac{{N - 1}}{{N + 1}}$.
The following conclusions hold.
\\(i) If the control parameters are set as
\begin{equation}
\begin{array}{l}
{\varepsilon _0}{ = }\frac{{6{\tilde r_2}}}{{N-1}},{\varepsilon _1} = 0,{\varepsilon _2} = \frac{{2{{({\tilde r_2})}^3}}}{{N-1}},{\theta _0} = \frac{{8{{({\tilde r_2})}^2}}}{{{{({\tilde r_2})}^2} + 3}},\\
{\theta _1} =  - 3{({\tilde r_2})^2},{\theta _2} = 3{({\tilde r_2})^2} \!-\! \frac{{8{{({\tilde r_2})}^2}}}{{{{({\tilde r_2})}^2} + 3}},
\end{array}
\end{equation}
the consensus is achieved with the convergence rate
\begin{equation}
\begin{aligned}
\tilde r_2=&  \mu\! +\! {\left( {{\mu^3} \!-\! \mu \!+\! \sqrt {{{({\mu^3} \!-\! \mu)}^2} \!\!+\! {{(1 \!-\! {\mu^2})}^3}} } \, \right)^{\!\frac{1}{3}}}
\\+ &\frac{{\sqrt 3 \rm{i} \!-\! 1}}{2}{\left( {{\mu^3} \!-\! \mu \!-\! \sqrt {{{({\mu^3} \!-\! \mu)}^2} \!\!+\! {{(1\! -\! {\mu^2})}^3} } } \,\right)^{\!\frac{1}{3}}}.
\end{aligned}
\end{equation}
\\(ii) The convergence rate of $M\!=\!2$ is better than the optimal convergence rate of $M\!=\!1$, i.e.,  $\tilde r_2<r_1^*$.
\end{theorem}

\begin{proof}
(i) For a star network with $N$ nodes, the eigenvalues of its Laplacian matrix are
\begin{equation} \nonumber
{\lambda _i} = \left\{ {\begin{array}{*{20}{l}}
{0,}&{\,\,\,\,i = 1}\\
{1,}&{\,\,\,\,2 \le i \le N - 1}\\
{N,}&{\,\,\,\,i = N}
\end{array}} \right..
\end{equation}
The consensus is achieved if and only if the roots of
\begin{equation}
\begin{aligned}
{p}(z,\lambda_i) =& {z^3} \!+\! ({\varepsilon _0}{\lambda _i} \!-\! 1 \!- \!{\theta _0})z{^2} \!+\! ({\varepsilon _1}{\lambda _i} \!+\! {\theta _0} \!+\! {\theta _2})z \\&+\! {\varepsilon _2}{\lambda _i} \!-\! {\theta _2} = 0,  \,\,\, i = 2,N,
\end{aligned}
\end{equation}
are in the unit circle.
Let $z = r\frac{{s + 1}}{{s - 1}}$ in (31), and denote
\begin{equation}
\begin{aligned}
&{f_0^{(i)}} \!= \!{r^3} \!+\! ({\varepsilon _0}{\lambda _i} \!-\! 1 \!-\! {\theta _0}){r^2} \!+\! ({\varepsilon _1}{\lambda _i} \!+\! {\theta _0} \!+\! {\theta _2})r \!+\! ({\varepsilon _2}{\lambda _i} \!-\! {\theta _2}),\\
&{f_1^{(i)}}\! = \!3{r^3} \!\!+\! ({\varepsilon _0}{\lambda _i} \!\!-\! 1 \!\!-\! {\theta _0}){r^2} \!\!-\! ({\varepsilon _1}{\lambda _i} \!\!+\! {\theta _0} \!+\! {\theta _2})r \!-\! 3({\varepsilon _2}{\lambda _i} \!\!-\! {\theta _2}),\\
&{f_2^{(i)}}\!=\! 3{r^3} \!\!-\!\! ({\varepsilon _0}{\lambda _i} \!\!-\! 1\! -\!\! {\theta _0}){r^2} \!\!-\!\! ({\varepsilon _1}{\lambda _i} \!\!+\! {\theta _0} \!+ {\theta _2})r \!+\! 3({\varepsilon _2}{\lambda _i} \!\!-\! {\theta _2}),\\
&{f_3^{(i)}} \!=\! {r^3} \!\!-\! ({\varepsilon _0}{\lambda _i}\! -\! 1 \!-\! {\theta _0}){r^2} \!+\! ({\varepsilon _1}{\lambda _i} \!+\! {\theta _0} \!+\! {\theta _2})r\!-\! ({\varepsilon _2}{\lambda _i}\! - \!{\theta _2}).
\end{aligned}
\end{equation}
Then the roots of $p(z,\lambda_i)\!=\!0$ are in the circle with radius $r$ if and only if the polynomial
\begin{equation}
{{\tilde p}}(s,\lambda_i)={f_0^{(i)}}s^3+{f_1^{(i)}}s^2+{f_2^{(i)}}s+{f_3^{(i)}}
\end{equation}
is stable.
Denote
${g^{(i)}}\!=\!{f_1^{(i)}}{f_2^{(i)}}\!-\!{f_0^{(i)}}{f_3^{(i)}}$.
According to the Routh stability criterion,
the polynomial (33) is stable or marginally stable, if and only if
\begin{equation}
f_0^{(i)} \ge 0,f_1^{(i)} \ge 0,f_2^{(i)} \ge 0, f_3^{(i)} \ge 0,g^{(i)} \ge 0,i = 2,N.
\end{equation}
Perform the operations
\begin{equation} \nonumber
\frac{3}{4}\! \times \!f_0^{(2)} \!+\! \frac{1}{2} \!\times\! f_1^{(2)} \!+\! \frac{1}{4} \!\times\! f_1^{(N)} \!+ \!\frac{1}{4} \!\times\! f_2^{(2)}\! + \! \frac{1}{2}\! \times\! f_2^{(N)} \!+\! \frac{3}{4} \!\times\! f_3^{(N)}\!,
\end{equation}
get $(1 \!-\! N){r^2}{\varepsilon _0} + 6{r^3} \ge 0$.
It follows that $\varepsilon _0 \le \frac{{6r}}{{N-1}}$.
The parameter $\varepsilon _0$ can be taken to the upper bound $\frac{{6r}}{{N-1}}$ when
\begin{equation}
f_0^{(2)} \!\!= 0, f_1^{(2)}\!\!= 0,  f_1^{(N)} \!\!= 0, f_2^{(2)} \!\!= 0, f_2^{(N)}\!\!= 0,  f_3^{(N)}\!\!= 0.
\end{equation}
Let the cubic polynomial
\begin{equation}
h(r) = {r^3} - 3\frac{{N - 1}}{{N + 1}}{r^2} + 3r - \frac{{N - 1}}{{N + 1}}.
\end{equation}
It is calculated by (35) that the control parameters satisfies (29) and $h(r) \!=\!0$.
Applying the Cardano's formula \cite{2010Cardano},
the sole real root of the equation $h(r) \!= \!0$ in the interval $r \in \left( {0,1} \right)$ is given by (30).
It is verified that the remaining constraints
\begin{equation} \nonumber
f_0^{(N)}>0,f_3^{(2)}> 0,g^{(2)} =g^{(N)}=0
\end{equation}
are satisfied by using the control parameters (29).
This completes the proof.
\\(ii)
Note that the derivative of the function $h(r)$ satisfies
\begin{equation} \nonumber
h'(r) =3{r^2} \!-\! 6\frac{{N - 1}}{{N + 1}}r \!+\! 3 > 3{r^2} \!- \!6r\! +\! 3 = 3{(r\! -\! 1)^2} > 0.
\end{equation}
Thus, $h(r)$ is increasing monotonically.
Substitute $r=r_1^*=\frac{{\sqrt N - 1}}{{\sqrt N + 1}}$ into $h(r)$, get
\begin{equation} \nonumber
h(r_1^*) = \frac{{8(\sqrt N  - 1)N}}{{{{(\sqrt N  + 1)}^3}(N + 1)}} > 0.
\end{equation}
It follows from $h(\tilde r_2) = 0<h(r_1^*)$ that $\tilde r_2 < r_1^*$.
\end{proof}

It follows from Theorem 3 that on star networks, the consensus with the two-tap memory can be achieved faster.
In fact, the consensus with the two-tap memory is not only accelerated on star networks.
For example, applying the control parameters (29), the maximum modulus root of $p(z,\lambda)=0,\lambda \in [1,5]$ is shown in Fig. 2.
It can be observed from Fig. 2 that the consensus with the two-tap memory is accelerated on some special networks, where the nonzero eigenvalues of its Laplacian matrix are in the interval $\lambda \in [1,1.28]\cup[4.72,5]$.

\begin{figure}[!htbp]
\centering
\includegraphics[height=6cm]{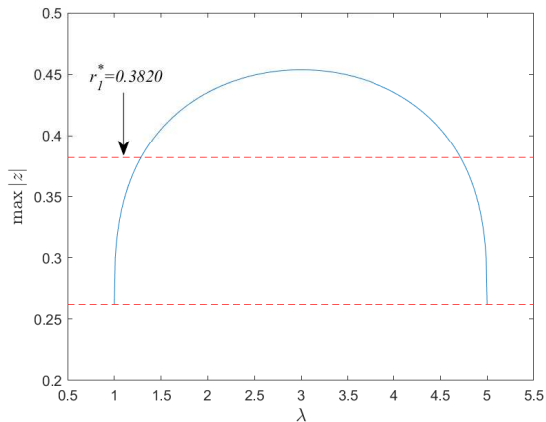}
\caption{The maximum modulus root of ${p}(z,\lambda)\!=\!0$}
\label{fig:label}
\end{figure}

\section{Numerical examples}
In this section, some examples are listed to verify the validity and correctness of the proposed results.

\begin{figure}[!htbp]
\centering
\includegraphics[height=6.2cm]{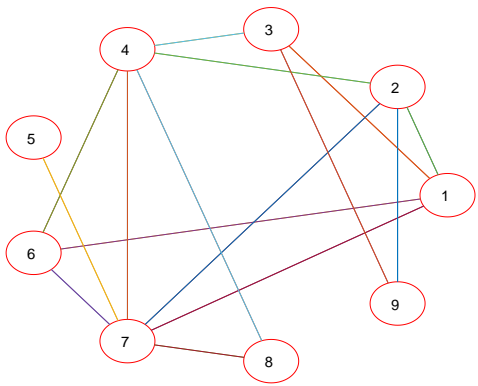}
\caption{The network topology $\mathcal{G}_1$}
\label{fig:label}
\end{figure}

\begin{example}
In this example, the convergence performance of different consensus algorithms with or without memory is compared.
The compared algorithms are:
(i) the best constant gain scheme (BC) proposed in \cite{XIAO200465},
(ii) accelerated consensus algorithm with the state deviation memory (SDMem) proposed in \cite{9117158},
(iii) the optimal one-tap memory scheme (OptMem) in this paper.
Randomly generate a network $\mathcal{G}_1$ with 9 nodes, as shown in Fig. 3.
It can be calculated that $\lambda_2\!=\!0.8835,\lambda_N\!=\!7.1716$.
Table I lists the optimal convergence rate and corresponding control parameters of each algorithm.
Generate the initial state of each agent in the interval $[-10,10]$.
In order to facilitate the comparison of the convergence speed,
the definition of $\epsilon$-convergence time \cite{ALEX2010CONVERGENCE} is introduced:
\[\mathcal{T}(\epsilon ) \!=\! \min \left\{ {{k^*}\! :\frac{{{{\left\| {\bm{x}(k) \! -\!  \bar x\bm{1}} \right\|}_\infty }}}{{{{\left\| {\bm{x}(0) \! -\!  \bar x\bm{1}} \right\|}_\infty }}} \le \epsilon   \,\,\,\, \forall k \! \ge\!  {k^*},\forall \bm{x}(0) \! \ne\!  \bar x\bm{1}} \right\}.\]
Set the error threshold as $\epsilon\!=\!10^{-5}$.
Fig. 4 shows the $\epsilon$-convergence time of each consensus algorithm.
It can be observed that the convergence speed of the optimal one-tap memory scheme proposed in this paper is faster than that of SDMem with $M\!=\!1$, and even faster than that of SDMem with $M\!=\!2$.
In addition, the convergence speed of the memoryless BC is the slowest.
\end{example}

\begin{table}[htbp]
\centering
  \caption{The optimal convergence rate and corresponding control parameters under different algorithms}
\begin{tabular}{ccccc}
\toprule
                  & OptMem            & BC \cite{XIAO200465} & SDMem \cite{9117158}  &  SDMem \cite{9117158}  \\
                  & with $M\!=\!1$    & without memory      & with $M\!=\!1$   &  with $M\!=\!2$  \\  \midrule
$r^*$             & 0.4804            & 0.7806               & 0.6402           & 0.5582\\ \midrule
$\varepsilon_0^*$ & 0.3056            & 0.2483               & 0.3180           & 0.3226\\ \midrule
$\varepsilon_1^*$ & 0                 & N/A                  & 0.0571           & 0.0789 \\ \midrule
$\varepsilon_2^*$ & N/A               & N/A                  & N/A              & 0.0112 \\ \midrule
$\theta_0^*$      & 0.2308            & N/A                  & N/A              &  N/A\\ \midrule
$\theta_1^*$      & -0.2308           & N/A                  & N/A              &  N/A\\
\bottomrule
\end{tabular}
\label{tab:addlabel}
\end{table}

\begin{figure}[!htbp]
\centering
\includegraphics[height=6.2cm]{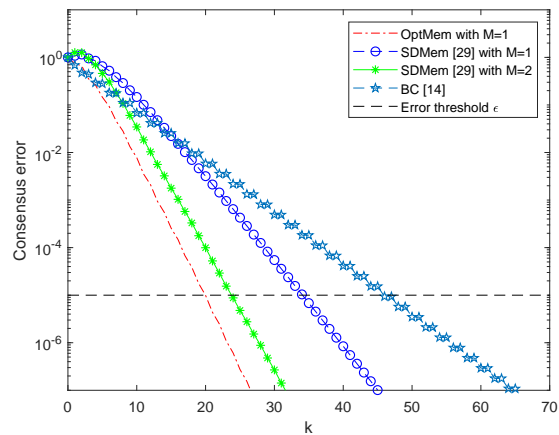}
\caption{The $\epsilon$-convergence time under different algorithms}
\label{fig:label}
\end{figure}

\begin{example}
In this example, the proposed algorithm is verified on large-scale networks.
Consider 80 random connected networks of 500 nodes,
which is randomly generated by a small-world network model under a rewiring probability $p\!=\!0.7$.
The optimal convergence rate on each network is shown in Fig. 5.
It can be seen from Fig. 5 that the algorithm proposed is also optimal on large-scale networks.
\end{example}

\begin{figure}[!htbp]
\centering
\includegraphics[height=6.2cm]{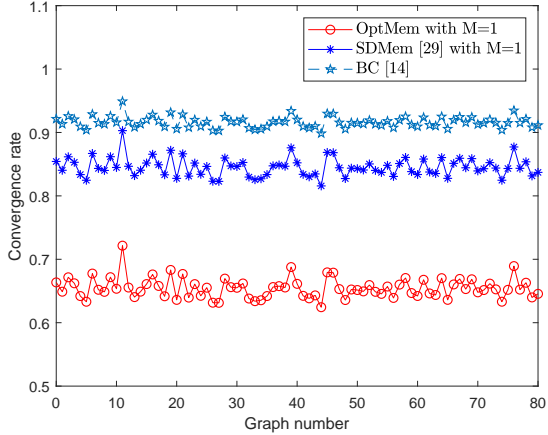}
\caption{The optimal convergence rate on large-scale networks}
\label{fig:label}
\end{figure}

%
%
%

\begin{example}
This example demonstrates the effectiveness of the control strategy in Theorem 3.
Table II lists the convergence rate $\tilde r_2$ and $r_1^*$ on the star network with different numbers of nodes.
Consider $N\!=\!10$, and randomly generate the initial state of each agent in the interval $[-10,10]$.
The convergence of consensus is shown in Fig. 6.
It can be observed that on the star network, the consensus with the two-tap memory scheme can be reached more quickly than that with the one-tap memory scheme.
\end{example}

\begin{table}[htbp]
\centering
  \caption{Convergence rate on the star network}
\begin{tabular}{cccccc}
\toprule
                 &$ N=5$           & $ N=10$      & $ N=20$   & $ N=50$  & $ N=100$ \\ \midrule
$r_1^*$          & 0.3820           & 0.5195      & 0.6345    & 0.7522   &  0.8182         \\
$\tilde r_2$       & 0.2620          & 0.3660     & 0.4616    & 0.5730   &  0.6455         \\
\bottomrule
\end{tabular}
\label{tab:addlabel}
\end{table}

\begin{figure}[!htbp]
\centering
\includegraphics[height=6.6cm]{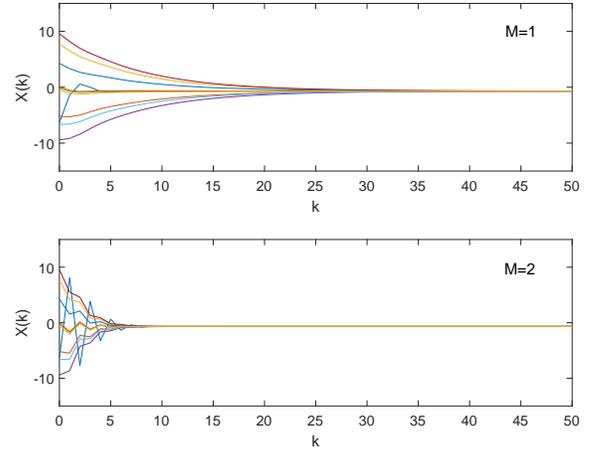}
\caption{Consensus on the star network}
\label{fig:label}
\end{figure}

\section{Conclusion}

This paper has proposed a more general control protocol with both the node memory and the state deviation memory to accelerate the consensus over multi-agent networks.
The optimal convergence rate with the one-tap memory has been formulated based on the Jury stability criterion.
The state deviation memory has been pointed out to be useless for the optimal convergence under the one-tap memory scheme.
By transforming the optimization problem of the convergence rate into the robust stabilization of the feedback system,
it has been proved that in the worst case, the one-tap node memory scheme is sufficient to achieve the optimal convergence, and adding any tap of the state deviation memory is unnecessary.
Moreover,
the two-tap state deviation memory has been found to be effective on some special networks.
Specially, for star networks, an optimized explicit convergence rate with the two-tap memory scheme has been given.
Numerical examples have demonstrated the validity and correctness of the obtained results.


%

%
%
%
%
%

\ifCLASSOPTIONcaptionsoff
  \newpage
\fi

\bibliographystyle{IEEEtran}
\bibliography{IEEEabrv,IEEEexample}

%

%
%
%




\end{document}